\documentclass[11pt,letterpaper]{article}

\usepackage[margin=1in]{geometry}
\usepackage{latexsym,graphicx,amssymb}
\usepackage{amsmath,enumerate}
\usepackage{bbm}
\usepackage{subfigure}
\usepackage{float}
\usepackage{epsfig}
\usepackage{xspace}
\usepackage{paralist}
\usepackage{enumerate}
\usepackage{cases}
\usepackage{caption}
\usepackage{multicol}
\usepackage{graphicx}
\usepackage{xcolor}
\usepackage{tikz}
\usepackage{ifthen}
\usepackage{algorithm}
\usepackage[noend]{algpseudocode}
\usepackage{amsthm}
\usepackage{tcolorbox}

\newtheorem{theorem}{Theorem}[section]

\newtheorem{lemma}{Lemma}[section]
\newtheorem*{lemma*}{Lemma}
\newtheorem{claim}{Claim}[section]
\newtheorem*{claim*}{Claim}

\newcommand{\be}{\begin{equation}}
\newcommand{\ee}{\end{equation}}
\newcommand{\beq}{\begin{equation*}}
\newcommand{\eeq}{\end{equation*}}

\newcommand{\AutoAdjust}[3]{\mathchoice{ \left #1 #2  \right #3}{#1 #2 #3}{#1 #2 #3}{#1 #2 #3} }
\newcommand{\Xcomment}[1]{{}}

\newcommand{\InBrackets}[1]{\AutoAdjust{[}{#1}{]}}\newcommand{\Ex}[2][]{\operatorname{\mathbf E}_{#1}\InBrackets{#2}}

\newcommand{\Prx}[2][]{\operatorname{\mathbf{Pr}}_{#1}\InBrackets{#2}}

\newcommand{\eqdef}{\overset{\mathrm{def}}{=\mathrel{\mkern-3mu}=}}
\newcommand{\vect}[1]{\ensuremath{\mathbf{#1}}}

\newcommand{\RN}[1]{\textup{\uppercase\expandafter{\romannumeral#1}}}

\newcommand\restr[2]{{\left.\kern-\nulldelimiterspace #1 \vphantom{\big|} \right|_{#2} }}
\def\prob{\Prx}

\def\expect{\Ex}

\newcommand{\alg}{\textsf{ALG}}
\newcommand{\opt}{\textsf{OPT}}

\newcommand{\dd}{\mathrm{d}}

\newcommand{\vals}{\vec{v}}

\newcommand{\dist}{\mathbf{D}}
\newcommand{\dists}{\vect{\dist}}

\title{Order Selection Prophet Inequality: \\
From Threshold Optimization to Arrival Time Design}
\author{
Bo Peng \thanks{ITCS, Shanghai University of Finance and Economics, \texttt{ahqspbo@gmail.com}}
\and
Zhihao Gavin Tang \thanks{ITCS, Shanghai University of Finance and Economics, \texttt{tang.zhihao@mail.shufe.edu.cn}}
}
\date{}

\begin{document}

\maketitle

\begin{abstract}
In the classical prophet inequality, a gambler faces a sequence of items, whose values are drawn independently from known distributions. Upon the arrival of each item, its value is realized and the gambler either accepts it and the game ends, or irrevocably rejects it and continues to the next item. The goal is to maximize the value of the selected item and compete against the expected maximum value of all items. 
A tight competitive ratio of $\frac{1}{2}$ is established in the classical setting and various relaxations have been proposed to surpass the barrier, including the i.i.d. model, the order selection model, and the random order model.

In this paper, we advance the study of the order selection prophet inequality, in which the gambler is given the extra power for selecting the arrival order of the items. 
Our main result is a $0.725$-competitive algorithm, that substantially improves the state-of-the-art $0.669$ ratio by Correa, Saona and Ziliotto~(Math. Program. 2021), achieved in the harder random order model.
Recently, Agrawal, Sethuraman and Zhang~(EC 2021) proved that the task of selecting the optimal order is NP-hard. 
Despite this fact, we introduce a novel algorithm design framework that translates the discrete order selection problem into a continuous arrival time design problem. From this perspective, we can focus on the arrival time design without worrying about the threshold optimization afterwards.
As a side result, we achieve the optimal $0.745$ competitive ratio by applying our algorithm to the i.i.d. model. \end{abstract}

\section{Introduction}
\label{sec:intro}
Prophet inequality has been a cornerstone of optimal stopping theory, since the classical result of Krengel and Sucheston~\cite{krengel1977semiamarts,krengel1978}. Consider a gambler facing a sequence of items, whose values are drawn independently from known distributions. After seeing an item, the gambler observes its realized value, and either accepts it and the game ends, or irrevocably rejects it and continues to the next item. The classical prophet inequality states that the gambler can achieve at least half of the expected maximum value. The latter is referred to as a prophet, who knows the realization of all values beforehand. Furthermore, the ratio of half is proven to be the best possible in the worst case.
Later, Samuel-Cahn~\cite{samuel1984} showed that the competitive ratio\footnote{We choose to use the terminology competitive ratio, due to the online nature of prophet inequality.} of $\frac{1}{2}$ can be achieved using a single-threshold algorithm.

In the past fifteen years, there has been an increased interest of prophet inequality related problems in the algorithmic game theory and online algorithms literature, due to its close connection to mechanism design and posted pricing mechanisms~\cite{aaai/HajiaghayiKS07}. Among the fruitful extensions of the classical prophet inequality, a remarkable line of research focuses on surpassing the $\frac{1}{2}$ impossibility result by relaxing the worst case model. Consider the following three variants in progressive order of difficulty.

\paragraph{I.I.D. Model.} Hill and Kertz~\cite{hill1982} studied the case when the value distributions are identical and designed a $1-\frac{1}{e}\approx 0.632$-competitive algorithm. They also constructed a family of instance showing that no algorithm can be better than $0.745$-competitive\footnote{The constant $\Gamma \approx 0.745$ is the unique solution to $\int_0^1 \frac{1}{y (1- \ln y) + 1/\Gamma -1} \dd y = 1$.}. The $1-\frac{1}{e}$ ratio is improved to $0.738$ by Abolhassan et al.~\cite{stoc/AbolhassaniEEHK17}. Recently, Correa et al.~\cite{mor/CorreaFHOV21} designed an optimal $0.745$-competitive algorithm, matching the hardness of Hill and Kertz.

\paragraph{Order Selection Model.} 
In this variant, the gambler is given an extra power for selecting the arrival order of each item. 
This assumption is natural in the application of sequential posted pricing mechanisms~\cite{stoc/ChawlaHMS10}, as the mechanism designer plays the role of the gambler.
Chawla et al.~\cite{stoc/ChawlaHMS10} proposed an $1-\frac{1}{e}$-competitive algorithm and the ratio is later improved to $1-\frac{1}{e}+0.022\approx 0.654$ by Beyhaghi et al.~\cite{or/BeyhaghiGLPS21}. 
This variant subsumes the i.i.d. model as a special case. Indeed, when the value distributions are identical, the extra power of order selection is useless.
Very recently, Agrawal, Sethuraman and Zhang~\cite{ec/AgrawalSZ20} established a negative result, showing that the task of selecting the optimal order is NP-hard, even when the support of each distribution is of size $3$. They also provide a $0.8$-competitive algorithm when the support of each distribution is of size at most $2$.

\paragraph{Random Order Model.} This variant is also known as prophet secretary, in which items arrive in a random order.
This model can be viewed as a generalization of the i.i.d. model and is no easier than the order selection model.
Esfandiari et al.~\cite{siamdm/EsfandiariHLM17} initiated the study of this variant and designed a $1-\frac{1}{e}$-competitive algorithm. 
Later, the same $1-\frac{1}{e}$ ratio is achieved using different strategies, including using personalized but time-invariant thresholds~\cite{mor/CorreaFHOV21}, and a single-threshold algorithm with randomized tie-breaking~\cite{soda/EhsaniHKS18}.
Later, the ratio is improved to $1-\frac{1}{e}+\frac{1}{400}$ by Azar, Chiplunkar, and Kaplan~\cite{ec/AzarCK18} and to $0.669$ by Correa, Saona, and Ziliotto~\cite{mp/CorreaSZ21}. The latter work also establishes a hardness of $\sqrt{3}-1 \approx 0.732$, showing a separation between the random order model and the i.i.d. model.

\subsection{Our Contributions.}
In this work, we focus on the order selection model.
Despite the NP-hardness of selecting the optimal order, we strongly exploit the power of order selection and design a $0.725$-competitive algorithm, that substantially improve the state-of-the-art $0.669$ ratio from the random order model. As a side result, our algorithm is $0.745$-competitive for the i.i.d. prophet inequality.

\paragraph{Previous Approaches.} 
We briefly summarize the previous techniques. 
Naturally, an algorithm is consisted of two parts: selecting the order and setting the thresholds. 
Each step is easy to optimize on its own. Specifically, when the arrival order is fixed, the optimal thresholds can be calculated through backward induction; when the thresholds are fixed for each item, we can calculate the expected value of each item conditioning on that its value exceeds the threshold, and then set the arrival order to be a descending order of the calculated values.

\begin{figure}
	\centering
	\includegraphics[width=\textwidth]{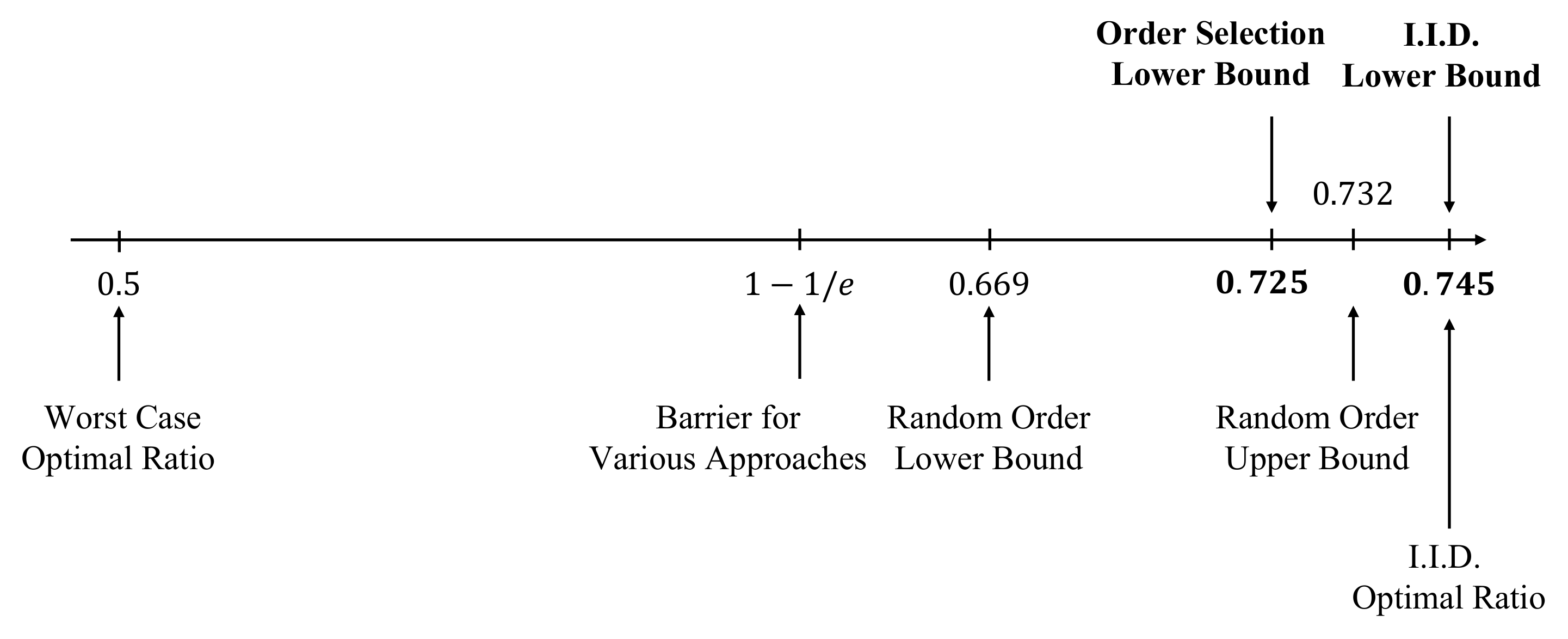}
	\caption{A summary of results. Our new results are marked in bold.
}
	\label{figure:results}
\end{figure}

Chawla et al.~\cite{stoc/ChawlaHMS10} and Beyhaghi et al.~\cite{or/BeyhaghiGLPS21} applied a two-step approach of first designing the thresholds, and then selecting the order. Both works studied the order selection prophet inequality from the perspective of sequential posted pricing mechanisms.
It is implicitly shown by Chawla et al.~\cite{stoc/ChawlaHMS10} that the latter setting reduces to the first setting.

The line of work studying prophet secretary~\cite{siamdm/EsfandiariHLM17, ec/AzarCK18, mp/CorreaSZ21} can be viewed as a two-step approach of first selecting the order, and then designing the thresholds. More accurately, the algorithm selects the uniform distribution over all permutations and then focuses on designing the thresholds. Remarkably, prior to our work, the state-of-the-art $0.669$ ratio for the order selection prophet inequality is established in the random order setting by Correa, Saona, and Ziliotto~\cite{mp/CorreaSZ21}.

\paragraph{Our Perspective: Arrival Time Design.} 
Recall a folklore continuous formulation of the prophet secretary problem. Let the time horizon be $[0,1]$ and assume that each item $i$ arrives at time $t_i \sim \textsf{Uni}[0,1]$ (i.e., the uniform distribution over $[0,1]$). This formulation is equivalent to the random arrival order and often eases the analysis. Specifically, under this formulation, Correa, Saona and Ziliotto~\cite{mp/CorreaSZ21} carefully set time-dependent thresholds and accept the first item whose value exceeds the threshold on its arrival time. 

We provide a novel point of view by re-scaling the time horizon. We first fix the time-dependent thresholds. Specifically, at time $t$, we set the threshold to be the value $\tau(t)$ so that the maximum value of all items is larger than it with probability exactly $t$. Then, we design an arrival time distribution $F_i$ for each item $i$ and let the items arrive at a random time with respect to $F_i$. 
In principle, this formulation is without loss of generality, since we can choose the distributions to be deterministic. 
Under this formulation, we only need to optimize for the arrival times. Moreover, the continuous formulation allows us to adapt the analysis framework from the i.i.d. setting~\cite{mor/CorreaFHOV21} and the random order setting~\cite{mp/CorreaSZ21}.
Noticeably, if the distributions $F_i$ are identical, our algorithm can be implemented in the prophet secretary setting. See Section~\ref{sec:prelim} for a more detailed discussion.

\paragraph{Our Results.}
We explicitly construct arrival time distributions and achieve a competitive ratio of $\Gamma=\frac{\ln \alpha+1}{\ln \alpha+1-\alpha} \approx 0.725$, where $\alpha  \approx 0.211 $ is the unique solution\footnote{For completeness, we provide a proof of the uniqueness of $\alpha$ in Appendix~\ref{app:missing}.} to $\int_{\alpha}^{1}  \frac{\ln \alpha+1}{ (\ln \alpha+1) \left( - x \cdot \ln x + x\right) - \alpha} \dd x +\frac{1}{\ln \alpha}=0$. 

Furthermore, our algorithm serves as an alternative optimal $\Gamma \approx 0.745$-competitive algorithm for the i.i.d. setting, with only one parameter modified, compared to our algorithm in the order selection setting. 
Our unified analysis bridges the i.i.d. setting and the order selection setting, and suggests that our novel \emph{arrival time design} perspective to be the right framework.

\subsection{Related Work}

There is a vast literature on prophet inequalities. We refer interested readers to the survey of Hill and Kertz~
\cite{hill1992survey} for the classical results, the suvery of Lucier~\cite{sigecom/Lucier17} for the economic perspective of prophet inequalities, and the survey of Correa et al.~\cite{sigecom/CorreaFHOV18} for more recent developments. Below, we review the most related works.

Hajiaghayi et al.~\cite{aaai/HajiaghayiKS07}, and Chawla et al.~\cite{stoc/ChawlaHMS10} observed a close relation between prophet inequalities and sequential posted pricing. They showed that designing posted pricing mechanisms can be reduced to the prophet inequality problem.
Recently, Correa et al.~\cite{DBLP:journals/orl/CorreaFPV19} proved that the two settings are indeed equivalent.

Besides the results that we have discussed before, there are a few special cases in which better competitive ratios are known for the order selection prophet inequality problem. When the number of items is a small constant, Beyhaghi et al.~\cite{or/BeyhaghiGLPS21} obtained a better competitive ratio than their general bound of $0.654$. If each type of distribution occurs at least $\Omega(\log n)$ times, Abolhassani et al.~\cite{stoc/AbolhassaniEEHK17} improved the competitive ratio to $0.738$ for the order selection model. Liu et al.~\cite{ec/LiuLPSS21} relaxed the problem by allowing the algorithm to remove a constant number of items. After so, they showed that the competitive ratio can be arbitrary close to $0.745$ against the relaxed prophet.

Closely related to the order selection prophet inequality is the optimal ordering problem. This problem shares the same input model as the order selection prophet inequality, while the benchmark is changed to the optimal online algorithm instead of the expected maximum value. Agrawal, Sethuraman and Zhang~\cite{ec/AgrawalSZ20} proved that the problem is NP-hard, and designed a FPTAS when the support of each distribution is of size $3$. Fu et al.~\cite{DBLP:conf/icalp/FuLX18} gave a PTAS when each distribution has a constant support size. Chakraborty et al.~\cite{DBLP:conf/wine/ChakrabortyEGMM10a} obtained a PTAS without any assumption on the support of the distribution. Their original results were stated in the setting of sequential posted pricing mechanisms, that can be translated to the optimal ordering problem by the reduction of Correa et al.~\cite{DBLP:journals/orl/CorreaFPV19}. 
Liu et al.~\cite{ec/LiuLPSS21} improved the results to an EPTAS based on a novel decomposition technique.

\section{Preliminaries}
\label{sec:prelim}
Let there be $n$ items, whose values $\vals = (v_1,v_2,\cdots,v_n)$ are drawn independently from known distributions $\dists = D_1 \times D_2 \times \cdots \times D_n$. The algorithm first selects an arrival order of the $n$ items. Then, the items arrive in a sequence according to the selected order. Upon the arrival of an item, its value is realized and the algorithm either accepts the item and stops, or rejects the item and continues to the next.
Our goal is to maximize the expected value of the selected item and compare against the prophet 
\[
\opt \eqdef \expect{\max_i v_i}.
\]

For the ease of presentation, we assume the value distributions $D_i$'s are continuous distributions. The extension to discrete distributions can be implemented by a careful tie-breaking rule. We refer to \cite{mp/CorreaSZ21} for a detailed explanation.

Our algorithm is parameterized by $n$ distributions $F_i$ for each $i \in [n]$, supported on $[0,1]$.
Consider the following algorithm:

\begin{tcolorbox}
\paragraph{Independent Arrival Time ($\{F_i\}$).}
\begin{itemize}
\item Sample independently $t_i \sim F_i$ for each $i$. We refer to $t_i$ as the arrival time of item $i$.
\item Let the items arrive in ascending order according to their arrival times.
\item We accept the first item $i$ with $v_i > \tau(t_i)$, where $\tau(t)$ is the threshold that 
\[\prob{\max_{i} v_i > \tau(t)} = t\]
\end{itemize}
\end{tcolorbox}

\paragraph{Remark.}
Before we go to the detailed analysis of our algorithm, it is worthwhile to make a comparison with the algorithm by Correa, Saona, and Ziliotto~\cite{mp/CorreaSZ21} for the \emph{prophet secretary} problem. In the prophet secretary problem (and other online optimization problems with random arrival), a folklore formulation is to assume that each item $i$ arrives at time $t_i \sim \textsf{Uni}[0,1]$ (i.e., the uniform distribution over $[0,1]$). Correa, Saona, and Ziliotto first set time-dependent thresholds $\tau(\alpha(t))$ at time $t$, with an appropriate function $\alpha$, and then accept the first item whose value exceeds the threshold. 

Alternatively, we re-scale the time horizon by fixing the threshold to be $\tau(t)$ at time $t$, and then let the items arrive according to carefully chosen distributions. 
Indeed, if all the distributions $F_i$'s are identical, our algorithm can be implemented in the prophet secretary setting.
Specifically, for any function $\alpha$, by setting $F_i(\alpha^{-1}(t))=t$ for every item $i$, our algorithm is equivalent to the algorithm of Correa, Saona, and Ziliotto~\cite{mp/CorreaSZ21}.
On the other hand, our formulation admits a natural generalization to the order selection setting by allowing non-identical $F_i$'s.

\paragraph{Analysis.}
Our analysis is similar to the framework of \cite{mp/CorreaSZ21}. We abuse $\alg$ to denote our algorithm and to denote the (random) value of the accepted item of our algorithm. We show the competitive ratio of our algorithm through the following stronger statement.

\begin{lemma}
\label{lem:main}
For the order selection prophet inequality, there exists distributions $\{F_i\}_{i \in [n]}$, so that for every $t\in [0,1]$:
\begin{equation*}
\prob{\alg > \tau(t)} \ge \Gamma \cdot t = \Gamma \cdot \prob{\max_i v_i > \tau(t)},
\end{equation*}
where $\Gamma=\frac{\ln \alpha+1}{\ln \alpha+1-\alpha} \approx 0.725$ and $\alpha \approx 0.211 $ is the unique solution to 
$\int_{\alpha}^{1}  \frac{\ln \alpha+1}{ (\ln \alpha+1) \left( - x \cdot \ln x + x\right) - \alpha} \dd x +\frac{1}{\ln \alpha}=0$.
\end{lemma}

Observe that for any non-negative random variable $V$, we have $\expect{V} = \int_0^\infty \prob{V > \tau} \dd \tau$. The above lemma immediately concludes the competitive ratio of our algorithm.

\begin{theorem}
The independent arrival time algorithm with functions $\{F_i\}$ chosen in Lemma~\ref{lem:main} is $\Gamma \approx 0.725$-competitive for the order selection prophet inequality. I.e., $\expect{\alg} \ge \Gamma \cdot \opt$.
\end{theorem}

As a side result, for the i.i.d. prophet inequality, i.e., when the distributions $D_1,D_2, \cdots, D_n$ are identical, our construction in Lemma~\ref{lem:main} works with a different parameter $\Gamma \approx 0.745$. Thus, we give an alternative optimal competitive algorithm for the i.i.d. prophet inequality. Formally, we prove the following lemma and theorem.

\begin{lemma}
\label{lem:iid}
For the i.i.d. prophet inequality, there exists distributions $\{F_i\}_{i \in [n]}$, so that for every $t\in [0,1]$:
\begin{equation*}
\prob{\alg > \tau(t)} \ge \Gamma \cdot t = \Gamma \cdot \prob{\max_i v_i > \tau(t)},
\end{equation*}
where $\Gamma \approx 0.745$ is the unique solution to $\int_0^1 \frac{1}{y (1- \ln y) + 1/\Gamma -1} \dd y = 1$.
\end{lemma}
\begin{theorem}
\label{thm:iid}
The independent arrival time algorithm with functions $\{F_i\}$ chosen in Lemma~\ref{lem:iid} is $\Gamma \approx 0.745$-competitive for the i.i.d. prophet inequality. I.e., $\expect{\alg} \ge \Gamma \cdot \opt$.
\end{theorem} 
\section{Analysis}
In this section, we prove Lemma~\ref{lem:main} and \ref{lem:iid}. We first provide the construction of the distributions $\{F_i\}$ in Section~\ref{sec:construction}, and then prove the stated inequality of Lemma~\ref{lem:main} and \ref{lem:iid} in Section~\ref{sec:competitive}.

Without specifying the constant $\Gamma$ and assuming that our distributions $\{F_i\}$ are well-defined, our constructions and analysis are unified for the non-i.i.d. case and the i.i.d. case.

Finally, in Section~\ref{sec:ratio}, we find the largest possible constants $\Gamma$ for our algorithm to be well-defined for the non-i.i.d. case and the i.i.d. case, respectively.
\subsection{Construction of $\{F_i\}$}
\label{sec:construction}
We explicitly construct the distributions $\{F_i\}_{i \in [n]}$ by defining their probability density functions. We first introduce some notations.
For every $t \in [0,1]$ and $i \in [n]$, let 
\[
p_i(t) \eqdef \prob{v_i > \tau(t)} \quad \text{and} \quad q_i(t) \eqdef \prob{\max_{j \ne i} v_i > \tau(t)}.
\]
With the assumption that the value distributions $\{D_i\}$ are continuous, we have that $p_i(t), q_i(t)$ are non-decreasing continuous functions. Hence, they are differentiable almost everywhere and we will use $p_i'(t), q_i'(t)$ to denote the derivatives. We have the following simple observation according to the definition of $p_i(t),q_i(t)$:
\begin{align}
& \prod_{i} (1-p_i(t)) = 1-t, \quad \forall t \in [0,1] \label{eqn:pt} \\
& \prod_{j \ne i} (1-p_j(t)) = 1-q_i(t), \quad \forall t \in [0,1], \forall i \in [n] \label{eqn:qt}
\end{align}
By taking derivatives on both sides of \eqref{eqn:pt}, we have
\begin{equation}
\label{eqn:derivative_p}
\sum_{i} p_i'(t) \cdot (1-q_i(t)) = \sum_{i} p_i'(t) \cdot \prod_{j \ne i}(1-p_j(t)) = 1.
\end{equation}

Consider equation~\eqref{eqn:pt} when $t=1$, we have $\prod_i (1-p_i(1)) = 0$. Hence, there exists at least one index $i$ with $p_i(1)=1$. Without loss of generality, let it be the index $1$. Consequently, $q_i(1) = 1 - \prod_{j \ne i} (1-q_j(1)) = 1$ for all $i \ne 1$.

Now, we define the distributions as the following.
\begin{tcolorbox}
\paragraph{Construction of $\{F_i\}$.}
\begin{itemize}
    \item Let $g(t) \eqdef \Gamma \cdot \left( \sum_i (1-q_i(t)) \cdot p_i(t) - t \right) + 1$ be an auxiliary function.
    \item For every item $i$, let $f_i(t) \eqdef \Gamma \cdot \frac{q_i'(t)}{g(t)} \cdot \exp \left(- \Gamma \cdot \int_0^t \frac{q_i'(s) \cdot p_i(s)}{g(s)} \dd s\right)$ be the probability density function for its arrival time $t_i \in [0,1)$ and let $t_i=1$ with probability $1 - \int_0^1 f_i(t) \dd t$. 
\end{itemize}
\end{tcolorbox}

\paragraph{Remark.}
We use almost the same construction for the non-i.i.d. case and the i.i.d. case, except for a different choice of the constant $\Gamma$. We shall specify the choice of $\Gamma$ later when it becomes crucial. 

We need to be careful when multiple items arrive at time $t=1$, since the distributions $F_i$'s might have point masses on $t_i=1$ according to our construction.
We resolve this issue by doing a special treatment for item $1$: 1) if $t_i=1$ for $i\ne 1$, we \emph{reject} item $i$ without looking at its realized value; 2) if $t_1 = 1$, we \emph{accept} it without looking at its realized value. Recall that $p_1(1)=\prob{v_1 > \tau(1)} =1$, there is no difference between always accepting item $1$ and setting a threshold of $\tau(1)$ to item $1$.

\paragraph{Intuition.} The construction might look mysterious at the first glance, with the complicated formulas. Indeed, our construction is driven by the analysis and is derived after solving a set of differential equations. We provide an informal argument in Appendix~\ref{app:informal} to provide intuitions how we derive the above distributions. For readers who are familiar with the analysis framework of Correa, Saona, and Ziliotto~\cite{mp/CorreaSZ21}, it would be helpful to check the informal argument before verifying the correctness of our proof. For other readers, we also encourage reading the informal argument after going through the full analysis. Nevertheless, our formal proof below is self-contained.

We first prove two useful mathematical properties of the functions $\{f_i(t)\}$ and $g(t)$.
\begin{lemma}
\label{lem:math_facts}
The functions $f_i(t), g(t)$ satisfy that
\begin{enumerate}
    \item $1-\int_0^t p_{i}(s) f_{i}(s) \dd s = \exp \left(- \Gamma \int_0^t \frac{q_i'(s) \cdot p_i(s)}{g(s)} \dd s\right)$, $\forall t \in [0,1], \forall i\in [n]$;
    \item $g(t) = \prod_{i} \left( 1 - \int_0^t f_i(s) p_i(s) \dd s \right)$, $\forall t \in [0,1]$.
\end{enumerate}
\end{lemma}
\begin{proof}
We verify the first equation by plugging in the definition of $f_i$ to the left hand side:
\begin{align*}
1-\int_0^t p_{i}(s) f_{i}(s) \dd s & = 1-\int_0^t p_{i}(s) \cdot \Gamma \cdot \frac{q_i'(s)}{g(s)} \cdot \exp \left(- \Gamma \cdot \int_0^s \frac{q_i'(x) \cdot p_i(x)}{g(x)} \dd x\right) \dd s \\
& = 1 + \int_{s=0}^t  1 \ \dd \left(\exp \left(- \Gamma \cdot \int_0^s \frac{q_i'(x) \cdot p_i(x)}{g(x)} \dd x\right) \right) \\
& = 1+ \left. \exp \left(- \Gamma \cdot \int_0^s \frac{q_i'(x) \cdot p_i(x)}{g(x)} \dd x\right) \right|_{s=0}^{t} \\
& = \exp \left(- \Gamma \cdot \int_0^t \frac{q_i'(s) \cdot p_i(s)}{g(s)} \dd s\right)
\end{align*}
Next, we prove the second statement. 
We first calculate the derivative of function $g(t)$:
\begin{multline}
\label{eqn:derivative_g}
g'(t) = \left(\Gamma \cdot \left( \sum_i (1-q_i(t)) \cdot p_i(t) - t \right) + 1 \right)' \\
= - \Gamma \cdot \sum_i q_i'(t) \cdot p_i(t) + \Gamma \cdot \sum_i (1-q_i(t)) \cdot p_i'(t) - \Gamma \overset{\eqref{eqn:derivative_p}}{=} - \Gamma \cdot \sum_i q_i'(t) \cdot p_i(t).
\end{multline}
Then, by applying the first stated equation, we have: 
\begin{multline*}
\prod_{i}\left( 1-\int_0^t p_{i}(s) f_{i}(s) \dd s \right) = \prod_{i} \exp \left(- \Gamma \cdot \int_0^t \frac{q_i'(s) \cdot p_i(s)}{g(s)} \dd s\right) \\
= \exp \left(\int_0^t \frac{- \Gamma \cdot \sum_{i}  q_i'(s) \cdot p_i(s)}{g(s)} \dd s\right) \overset{\eqref{eqn:derivative_g}}{=} \exp \left(\int_0^t \frac{g'(s)}{g(s)} \dd s\right) = \exp \left( \left. \ln g(s) \right|_{s=0}^t \right) = \frac{g(t)}{g(0)} = g(t).
\end{multline*}
\end{proof} 
\subsection{Competitive Analysis}
\label{sec:competitive}
For our algorithm to be well-defined, we need to verify that the distributions are valid, i.e. $\int_0^1 f_i(t) \dd t \le 1$ for all $i \in [n]$. This is the crucial place where we have different constants $\Gamma$ for the non-i.i.d. case and i.i.d. case respectively. 

\begin{lemma}
\label{lem:ratio_iid}
For the i.i.d. case, for $\Gamma \approx 0.745$ (the unique solution to $\int_0^1 \frac{1}{y (1- \ln y) + 1/\Gamma -1} \dd y = 1$), and for each $i \in [n]$, we have 
\[
\int_0^1 f_i(t) \dd t \le 1.
\]
\end{lemma}

For the non-i.i.d. case, we prove the following stronger statement that automatically implies the validity of our algorithm.
\begin{lemma}
\label{lem:ratio}
For the non-i.i.d. case, for $\Gamma=\frac{\ln \alpha+1}{\ln \alpha+1-\alpha} \approx 0.725$ where $\alpha \approx 0.211 $ is the unique solution to 
$\int_{\alpha}^{1}  \frac{\ln \alpha+1}{ (\ln \alpha+1) \left( - x \cdot \ln x + x\right) - \alpha} \dd x +\frac{1}{\ln \alpha}=0$, and for each $i \in [n]$, we have
\[
(1-\Gamma \cdot q_{i}(1)) \cdot \left( 1 - \int_0^1 f_i(t) \dd t \right) \cdot \exp \left(\Gamma \int_0^1 \frac{q_i'(s) \cdot p_i(s)}{g(s)} \dd s\right) \ge \Gamma \cdot (1-q_i(1)).
\]
\end{lemma}

We defer the proofs of the above lemmas to the next subsection and continue proving the stated inequality of Lemma~\ref{lem:main} and \ref{lem:iid}, assuming the validity of our algorithm.

Fixing an arbitrary time $t \in [0,1]$, the event that our algorithm accepts an item with value larger than $\tau(t)$ belongs to one of the following $n+1$ possibilities:
\begin{itemize}
    \item Our algorithm stops before time $t$. In this case, the accepted value must be larger than $\tau(t)$, since the threshold function $\tau$ is decreasing.
    \item For some $i \in [n]$, our algorithms accepts item $i$ at time $t_i \ge t$ and $v_i > \tau(t)$.
\end{itemize}

We introduce notations $A_i(t), B_i(t)$ for each $i\in [n]$ to denote the following events:
\begin{itemize}
 \item $A_i(t)$: item $i$ arrives at time $t_i < t$ and $v_i > \tau(t_i)$.
 \item $B_i(t)$: item $i$ is accepted by our algorithm at time $t_i \ge t$ and $v_i > \tau(t)$.
\end{itemize}

\begin{lemma}
\label{lem:before_t}
For any $t \in [0,1]$, $\prob{\alg \text{ stops before time } t} = 1- g(t)$.
\end{lemma}
\begin{proof}
Observe that our algorithm stops before time $t$ if and only if at least one of the events $\{A_i(t)\}_{i}$ happens. Moreover, the events $A_i(t)$ are independent from each other for different $i$'s. Consequently,
\begin{multline*}
\prob{\alg \text{ stops before time } t} = \prob{\cup_i A_i(t)} = 1 - \prod_{i} \left(1-\prob{A_i(t)} \right) \\
= 1 - \prod_i \left( 1 - \int_0^t p_i(t_i) \cdot f_i(t_i) \dd t_i \right) = 1 - g(t),
\end{multline*}
where the last equation follows from the second statement of Lemma~\ref{lem:math_facts}.
\end{proof}

Next, we study the events $\{B_i(t)\}$. 
\begin{lemma}
\label{lem:after_t}
For any $t \in [0,1]$ and $i \in [n]$, $\prob{B_i(t)} \ge \Gamma \cdot p_i(t) \cdot (1-q_i(t))$.
\end{lemma}
\begin{proof}
For any $i$, fixing the arrival time $t_i \in (t,1)$ of $i$ and conditioning on that its realized value $v_i$ is larger than $\tau(t)$, our algorithm accepts it as long as we haven't stopped before time $t_i$. Specifically, the last event happens when none of the $\{A_j(t_i)\}_{j \ne i}$ happens. Thus,
\begin{align*}
\prob{B_i(t)} & \ge \int_t^1 f_i(t_i) \cdot \prob{v_i > \tau(t)} \cdot \prob{i \text{ is accepted by } \alg \mid t_i, v_i > \tau(t)} \dd t_i \\
& = \int_t^1 f_i(t_i) \cdot p_i(t) \cdot \prod_{j \ne i} \left( 1 - \prob{A_j(t_i)} \right) \dd t_i \\
& = p_i(t) \cdot \int_t^1 f_i(t_i) \cdot \prod_{j \ne i} \left( 1 - \int_0^{t_i} p_j(t_j) \cdot f_j(t_j) \dd t_j \right) \dd t_i \\
& = p_i(t) \cdot \int_t^1 f_i(t_i) \cdot \frac{\prod_{j} \left( 1 - \int_0^{t_i} p_j(t_j) \cdot f_j(t_j) \dd t_j \right)}{1 - \int_0^{t_i} p_i(s) f_i(s) \dd s} \dd t_i \\
& = p_i(t) \cdot \int_t^1 f_i(t_i) \cdot \frac{g(t_i)}{\exp \left(-\Gamma \int_0^{t_i} \frac{q_i'(s) \cdot p_i(s)}{g(s)} \dd s\right)} \dd t_i \tag{by Lemma~\ref{lem:math_facts}} \\
& = p_i(t) \cdot \int_t^1 f_i(t_i) \cdot g(t_i) \cdot \exp \left(\Gamma \int_0^{t_i} \frac{q_i'(s) \cdot p_i(s)}{g(s)} \dd s\right) \dd t_i \\
& = p_i(t) \int_t^1 \Gamma \cdot q_i'(t_i) \dd t_i \tag{by the definition of $f_i$}\\
& = \Gamma \cdot p_i(t) \cdot (q_i(1) - q_i(t)) 
\end{align*}
For $i \ne 1$, we conclude the proof of the statement by noticing that $q_i(1)=1$.
However, $q_1(1)$ not necessarily equals $1$.
We remark that for the i.i.d. case, all distributions are symmetric and the above analysis is sufficient since $q_i(1)=1$ for all $i \in [n]$. The rest of our proof is only for the non-i.i.d. case.

Note that the above analysis ignores the point mass of $F_1$ on $t_1=1$ and recall that we have a special treatment of item $1$ when it arrives at time $1$. 
It suffices to calculate the extra probability when item $1$ is accepted at time $1$ and $v_1 > \tau(t)$.
\begin{align*}
& \prob{\alg \text{ accepts item 1 at time 1 and } v_1 > \tau(t)} \\
& = \prob{t_1 = 1} \cdot \prob{v_1 > \tau(t)} \cdot \prod_{j \ne 1} \left( 1 - \prob{A_j(1)} \right) \\
& = \left( 1 - \int_0^1 f_1(t_1) \dd t_1 \right) \cdot p_1(t) \cdot \left( g(1) \cdot \exp \left(\Gamma \int_0^1 \frac{q_1'(s) \cdot p_1(s)}{g(s)} \dd s\right) \right) \\
& = (1-\Gamma \cdot q_{1}(1)) \cdot p_1(t) \cdot \left( 1 - \int_0^1 f_1(t_1) \dd t_1 \right) \cdot \exp \left(\Gamma \int_0^1 \frac{q_1'(s) \cdot p_1(s)}{g(s)} \dd s\right) \\
& \ge \Gamma \cdot p_1(t) \cdot (1-q_1(1)).
\end{align*}
Here, the third equality follows from the fact that $g(1) = \Gamma \cdot \left( \sum_i (1-q_i(1)) \cdot p_i(1) - 1 \right) + 1 = 1 - \Gamma \cdot q_{1}(1)$; the last inequality follows from Lemma~\ref{lem:ratio}. This concludes the proof of the lemma.
\end{proof}

Equipped with the above lemmas, we conclude the proof of Lemma~\ref{lem:main} and \ref{lem:iid}.
\begin{align*}
\prob{\alg > \tau(t)} & = \prob{\alg \text{ stops before time } t} + \sum_{i} \prob{B_i(t)} \\
& \ge 1- g(t) + \sum_{i} \Gamma \cdot p_i(t) \cdot (1-q_i(t)) \tag{by Lemma~\ref{lem:before_t} and \ref{lem:after_t}}\\
& = 1 - \left( \Gamma \cdot \left( \sum_i (1-q_i(t)) \cdot p_i(t) - t \right) + 1 \right) + \Gamma \cdot \sum_i p_i(t) \cdot (1-q_i(t)) \\
& = \Gamma \cdot t~.
\end{align*}

\subsection{Calculation of $\Gamma$}
\label{sec:ratio}
Finally, we prove Lemma~\ref{lem:ratio_iid} and \ref{lem:ratio}. Recall the definition of $f_i(t)$. We have that
\[
\int_0^1 f_i(t) \dd t = \int_0^1  \frac{\Gamma \cdot q_i'(t)}{g(t)  \exp \left( \Gamma \cdot \int_0^t \frac{q_i'(s) \cdot p_i(s)}{g(s)} \dd s\right)}  \dd t 
\]
Since $q_i(t)$ is continuous and non-decreasing, we do the following change of variables: for $x \in [0,q_{i}(1)]$,
\begin{itemize}
    \item $q_{i}^{-1}(x) \eqdef \sup \left\{t \mid q_{i}(t) \le x \right\} $;
    \item $\tilde{p}_i(x) \eqdef p_{i}(q_{i}^{-1}(x))$ and $\tilde{g}_i(x) \eqdef g (q_{i}^{-1}(x))$. 
\end{itemize}
Then, 
\[
\int_0^1 f_i(t) \dd t = \int_{t=0}^1  \frac{\Gamma}{g(t)  \exp \left( \Gamma \cdot \int_{s=0}^t \frac{p_i(s)}{g(s)} \dd q_i(s)\right)}  \dd q_i(t) = \int_{0}^{q_i(1)} \frac{\Gamma}{\tilde{g}_i(x)  \exp \left( \Gamma \cdot \int_{0}^x \frac{\tilde{p}_i(y)}{\tilde{g}_i(y)} \dd y \right)}  \dd x~.
\]
\paragraph{Remark.}
If $q_i(t)$ is strictly monotonically increasing, our definition of $q_{i}^{-1}(x)$ is the standard inverse function of $q_i(t)$. 
The above change of variables works for arbitrary absolute continuous non-decreasing function $q_i(t)$. Indeed, for any Lebesgue measurable function $h \ge 0$ and $a\le b$, 
\[
\int_{a}^{b} h(t) q_i'(t) \dd t = \int_{a}^{b} h(q_i^{-1}(q_i(t))) q_i'(t) \dd t = \int_{q_i(a)}^{q_i(b)} h(q_i^{-1}(x)) \dd x,
\]
where the first equation follows from the fact that the Lebesgue measure of $\{x \mid h(q_i^{-1}(q_i(t))) q_i'(t) \ne  h(t) q_i'(t)\}$ equals $0$.

\subsubsection{I.I.D.: Proof of Lemma~\ref{lem:ratio_iid}}
We start with the case of i.i.d. distributions.
Within this subsection, $\Gamma \approx 0.745$ is the unique solution to $\int_0^1 \frac{1}{y (1- \ln y) + 1/\Gamma -1} \dd y = 1$.
By symmetry, all functions $p_i(t)$ are the same. Since $\prod_i (1- p_i(t)) = 1-t$, we have that for all $i$
\[
p_{i}(t) = 1-(1-t)^{\frac{1}{n}} \quad \text{and} \quad q_{i}(t) = 1 - \prod_{j \ne i} (1 - p_j(t)) = 1-(1-t)^{\frac{n-1}{n}}
\]
Consequently, we have 
\begin{align*}
& \tilde{p_{i}}(x)=1-(1-q_i^{-1}(x))^{\frac{1}{n}} = 1 - (1-x)^{\frac{1}{n-1}} \\
\text{and} \quad & \tilde{g}_i(x)=\Gamma \cdot \left( \sum_j (1-q_j(q_i^{-1}(x))) \cdot p_j(q_i^{-1}(x)) - q_i^{-1}(x) \right) + 1 \\
& \phantom{\tilde{g}_i(x)} = \Gamma \cdot \left( \sum_j (1-x) \cdot \tilde{p}_i(x) - q_i^{-1}(x) \right) + 1 \tag{since $p_i,q_i$ are the same for all $i$}\\
& \phantom{\tilde{g}_i(x)} = \Gamma \cdot \left( n \left( 1-x -(1-x)^{\frac{n}{n-1}}\right) -1+ (1-x)^{\frac{n}{n-1}} \right)+1
\end{align*}

We have the following mathematical fact, whose proof involves tedious calculations that we defer to Appendix~\ref{app:missing}.
\begin{claim}
\label{clm:iid}
For any $x \in [0,1]$, we have
\begin{equation*}
\tilde{g}_i(x) \cdot \exp \left( \Gamma \cdot \int_{0}^x \frac{\tilde{p}_i(y)}{\tilde{g}_i(y)} \dd y \right) \ge \Gamma \cdot \left( -(1-x) \ln(1-x)-x\right)+1
\end{equation*}
\end{claim}

Applying the above claim and recalling that $q_i(1)=1$, we have that
\begin{align*}
\int_0^1 f_{i}(t) \dd t & =  \int_0^{1}  \frac{\Gamma}{\tilde{g}_i(x) \cdot \exp \left( \Gamma \cdot \int_{0}^x \frac{\tilde{p}_i(y)}{\tilde{g}_i(y)} \dd y \right)} \dd x \le \int_0^{1}  \frac{\Gamma }{ \Gamma \cdot \left( -( 1-x) \ln(1-x)-x\right)+1} \dd x~ \\
& = \int_0^1 \frac{\Gamma}{\Gamma \cdot \left( -y \cdot \ln y - (1-y) \right) + 1} \dd y \tag{let $y=1-x$}\\
& = \int_0^1 \frac{1}{y (1-\ln y) + \frac{1}{\Gamma} - 1} \dd y = 1,
\end{align*}
where the last equality follows from the definition of $\Gamma$.

\subsubsection{Non-I.I.D.: Proof of Lemma~\ref{lem:ratio}}
Finally, we derive the constant $\Gamma$ for the non-i.i.d. case. In contrast to the analysis for the i.i.d. case, we no longer have explicit expressions for functions $p_i(t),q_i(t)$. The challenge is to prove that for all possible $p_i(t), q_i(t)$, the stated inequality holds.
Within this subsection, $\Gamma=\frac{\ln \alpha+1}{\ln \alpha+1-\alpha}  \approx 0.725$ and $\alpha \approx 0.211 $ is the unique solution of the following equation on $(0,1)$
\[
\int_{\alpha}^{1}  \frac{\ln \alpha+1}{ (\ln \alpha+1) \left( - x \cdot \ln x + x\right) -\alpha} \dd x +\frac{1}{\ln \alpha}=0.
\]

We first observe the following property regarding functions $\tilde{p}_i(x)$ and $\tilde{g}_i(x)$. 
\begin{claim}
\label{clm:gp}
For each $x \in [0, q_i(1)]$, we have
\begin{equation}
\label{eqn:gp}
    \tilde{g}_i(x) \ge \Gamma \cdot \left( - (1-x) \cdot \ln(1-x) \cdot (1-\tilde{p}_i(x)) - x \right) + 1
\end{equation}
\end{claim}
\begin{proof}
For notation simplicity, let $t = q_i^{-1}(x)$. 
Since $(1-q_i(t)) \cdot (1-p_i(t)) = 1-t$,  we have that $p_i(t) = 1 - \frac{1-t}{1-q_i(t)} = \frac{x-t}{1-x}$. Then,
\begin{align*}
\tilde{g}_i(x) & =\Gamma \cdot \left( \sum_j (1-q_j(t)) \cdot p_j(t) - t \right) + 1 = \Gamma \cdot \left( \sum_{j \ne i} (1-q_j(t)) \cdot p_j(t) + (1-q_i(t)) \cdot p_i(t) - t \right) + 1 \\
& = \Gamma \cdot \left( (1-t) \cdot \sum_{j\ne i} \frac{p_j(t)}{1-p_j(t)} + (1-x) \cdot \frac{x-t}{1-x} - t \right) + 1 \\
& \ge \Gamma \cdot \left( (1-t) \cdot \sum_{j\ne i} \ln \left( \frac{1}{1-p_j(t)} \right) - x \right) + 1 = \Gamma \cdot \left( (1-t) \cdot \ln \left( \frac{1}{\prod_{j \ne i} (1-p_j(t))} \right) - x \right) + 1 \\
& = \Gamma \cdot \left( (1-t) \cdot \ln \left( \frac{1}{1-q_i(t)} \right) - x \right) + 1 \\
& = \Gamma \cdot \left( - (1-q_i(t)) \cdot (1-p_i(t)) \cdot \ln \left(1-q_i(t) \right) - x \right) + 1 \\
& = \Gamma \cdot \left( - (1-x) \cdot \ln(1-x) \cdot (1-\tilde{p}_i(x)) - x \right) + 1
\end{align*}
Here, the third equality holds since $(1-q_j(t))\cdot (1-p_j(t)) = 1-t$ for all $j$; the inequality holds since $\frac{p}{1-p} \ge \ln \frac{1}{1-p}$ for all $p \in [0,1)$.
\end{proof}

Observe that functions $\tilde{p}_i(x), \tilde{g}_i(x)$ are only defined on $[0,q_i(1)]$. We further extend the two functions by defining $\tilde{p}_i(x) = 1$ and $\tilde{g}_i(x) = 1 - \Gamma \cdot x$ for $x \in (q_i(1),1]$. It is straightforward to verify that the extended functions satisfy \eqref{eqn:gp} for all $x \in [0,1]$.
This condition is the only property that we are going to use for functions $\tilde{p}_i$ and $\tilde{g}_i$. Specifically, we prove the following technical lemma.
\begin{lemma}
\label{lem:integral}
Suppose functions $\tilde{p}, \tilde{g}:[0,1] \to [0,1]$ satisfies that 
\[
\tilde{g}(x) \ge \Gamma \cdot \left( - (1-x) \cdot \ln(1-x) \cdot (1-\tilde{p}(x)) - x \right) + 1.
\]
Then
\[
\int_0^{1}  \frac{\Gamma}{\tilde{g}(x)  \exp \left( \Gamma \cdot \int_0^{x} \frac{ \tilde{p}(y)}{\tilde{g}(y)} \dd y\right)} \dd x \le 1.
\]
\end{lemma}
We defer the proof of the lemma to Appendix~\ref{app:missing}. By applying it to the extended functions $\tilde{p}_i(x), \tilde{g}_i(x)$, we conclude the proof of Lemma~\ref{lem:ratio}:
\begin{align*}
& \phantom{=} (1-\Gamma \cdot q_{i}(1)) \cdot \left( 1 - \int_0^1 f_i(t) \dd t \right) \cdot \exp \left(\Gamma \int_0^1 \frac{q_i'(s) \cdot p_i(s)}{g(s)} \dd s\right) \\ 
& = (1- \Gamma \cdot q_{i}(1)) \cdot \left( 1- \int_0^{q_{i}(1)}  \frac{\Gamma}{\tilde{g}_i(x)  \exp \left( \Gamma \cdot \int_0^{x} \frac{ \tilde{p}_i(y)}{\tilde{g}_i(y)} \dd y\right)} \dd x \right) \exp {\left(\Gamma \int_0^{q_i(1)} \frac{ \tilde{p}_i(x)}{\tilde{g}_i(x)} \dd x\right)} \\
& \ge (1- \Gamma \cdot q_{i}(1)) \cdot \left( \int_{q_{i}(1)}^{1}  \frac{\Gamma}{\tilde{g}_i(x)  \exp \left( \Gamma \cdot \int_0^{x} \frac{ \tilde{p}_i(y)}{\tilde{g}_i(y)} \dd y\right)} \dd x \right) \exp {\left(\Gamma \int_0^{q_i(1)} \frac{ \tilde{p}_i(x)}{\tilde{g}_i(x)} \dd x\right)} \tag{by Lemma~\ref{lem:integral}}\\
& = (1- \Gamma \cdot q_{i}(1)) \cdot \int_{q_{i}(1)}^{1}  \frac{\Gamma}{(1-\Gamma \cdot x)  \exp \left( \Gamma \cdot \int_{q_i(1)}^{x} \frac{ 1}{1-\Gamma \cdot y} \dd y\right)} \dd x  \tag{by our extension of $\tilde{p}_i,\tilde{g}_i$}\\
& = (1- \Gamma \cdot q_{i}(1)) \cdot  \int_{q_{i}(1)}^{1}  \frac{\Gamma}{(1-\Gamma \cdot x)  \exp \left( \left. -\ln(1-\Gamma \cdot y) \right|_{q_i(1)}^{x} \right)} \dd x  \\
& = (1- \Gamma \cdot q_{i}(1)) \cdot \int_{q_{i}(1)}^{1}  \frac{\Gamma}{1-\Gamma \cdot q_i(1)} \dd x \\
& = \Gamma \cdot (1-q_i(1))~.
\end{align*}

\bibliography{prophet}
\bibliographystyle{plain}

\newpage

\appendix

\section{Informal Arguments}
\label{app:informal}
For readers who are familiar with the analysis of Correa, Saona, and Ziliotto~\cite{mp/CorreaSZ21}, or readers who have finished reading Section~\ref{sec:competitive}, we explain how the arrival time distributions are designed in this section. 
Specifically, we start with our analysis without specifying the distributions $F_i$.
\begin{align*}
\Pr[\alg \ge \tau(t)] & \ge 1 - \prod_{i} \left( 1 - \int_0^t p_i(t_i) f_i(t_i) \dd t_i \right) \\
& + \sum_{i} \left( p_i(t) \cdot \int_t^1 \prod_{j\ne i} \left( 1 - \int_0^{t_i} p_j(t_j) f_j(t_j) \dd t_j \right) f_i(t_i) \dd t_i \right)
\end{align*}
Here, the first line corresponds to the case when our algorithm stops before time $t$ (refer to Lemma~\ref{lem:before_t}) and the second line corresponds to the case when our algorithm stops after time $t$ (refer to Lemma~\ref{lem:after_t}).

We aim at designing functions $\{f_i\}$ so that the right hand side equals $\Gamma \cdot t$ for every $t \in [0,1]$. Denote the right hand side by $H(t)$. Then we have,
\begin{align*}
H'(t) & = \sum_{i} p_i(t) f_i(t) \cdot \prod_{j \ne i} \left( 1 - \int_0^t p_i(t_j) f_j(t_j) \dd t_j \right) - \sum_{i} p_i(t) f_i(t) \cdot \prod_{j \ne i} \left( 1 - \int_0^t p_j(t_j) f_j(t_j) \dd t_j \right) \\
& + \sum_{i}  \left( p_i'(t) \cdot \int_t^1  \prod_{j\ne i} \left( 1 - \int_0^{t_i} p_j(t_j) f_j(t_j)\ dd t_j \right) f_i(t_i) \dd t_i  \right)\\
& = \sum_{i}  \left( p_i'(t) \cdot \int_t^1  \prod_{j\ne i} \left( 1 - \int_0^{t_i} p_j(t_j) f_j(t_j) \dd t_j \right) f_i(t_i) \dd t_i  \right)
\end{align*}

Recall equation~\eqref{eqn:derivative_p}, $\sum_{i} p_i'(t) \cdot (1-q_i(t)) = 1$ for all $t$.
If we set $\{f_i\}$ to satisfy the following equations, we shall automatically have that $H'(t) = \Gamma$ for all $t \in [0,1]$. 
\begin{align*}
& \int_t^1  \prod_{j\ne i} \left( 1 - \int_0^{t_i} p_j(t_j) f_j(t_j) \dd t_j \right) f_i(t_i) \dd t_i = \Gamma \cdot (1-q_i(t)), \quad \forall i \in [n] \\
\iff & \prod_{j\ne i} \left( 1 - \int_0^{t_i} p_j(t_j) f_j(t_j) \dd t_j \right) f_i(t_i) = - \Gamma \cdot q_i'(t), \quad \forall i \in [n]
\end{align*}
We remark that this is only a sufficient but not necessary condition for our algorithm to work. We then solve the above set of differential equations and achieve our constructions of $\{f_i\}$.

\section{Missing Proofs}
\label{app:missing}
\subsection{Uniqueness of $\alpha$}
\begin{lemma*}
There exists a unique $\alpha \in (0,1)$, satisfying the following equation
\begin{equation*}
    \int_{\alpha}^{1}  \frac{\ln \alpha+1}{ (\ln \alpha+1) \left( - x \cdot \ln x + x\right) -\alpha} \dd x +\frac{1}{\ln \alpha}=0.
\end{equation*}
\end{lemma*}
\begin{proof}
Denote
\[
Y(z) \eqdef \int_{z}^{1}  \frac{\ln z+1}{ (\ln z+1) \left( - x \cdot \ln x + x\right) -z} \dd x +\frac{1}{\ln z}.
\]
We calculate the derivative of $Y(z)$.
\begin{align*}
Y'(z) & = \int_{z}^{1} \frac{\frac{1}{z}((\ln z +1)(-x\ln x +x)-z)-(\ln z +1)\left(\frac{-x\ln x+x}{z}-1\right)}{((\ln z+1)(-x\ln x+x)-z)^{2}} \dd x -\frac{\ln z +1}{(\ln z +1)(-z \ln z+z)-z}\\
& \phantom{=} -\frac{1}{z(\ln z)^{2}}\\
& = \frac{1}{z \ln z}+ \int_{z}^{1} \frac{\ln z}{((\ln z+1)(-x\ln x+x)-z)^{2}} \dd x < 0,
\end{align*}
where the last inequality holds by $\ln z < 0$ for $z \in (0,1)$. Therefore, we obtain the monotonicity of $Y(z)$ on $(0,1)$. Since $Y(0^{+}) > 0$ and $Y(1^{-}) < 0$, there exists the unique $\alpha \in (0,1)$ such that $Y(\alpha) = 0$. Moreover, our numerical result shows that $\alpha \approx 0.211$.
\end{proof}

\subsection{Proof of Claim~\ref{clm:iid}} 
We restate the claim as the following.
\begin{claim*}
For any $x \in [0,1]$, we have
\begin{equation*}
\tilde{g}_i(x) \cdot \exp \left( \Gamma \cdot \int_{0}^x \frac{\tilde{p}_i(y)}{\tilde{g}_i(y)} \dd y \right) \ge \Gamma \cdot \left( -(1-x) \ln(1-x)-x\right)+1~,
\end{equation*}
where $\tilde{p_{i}}(x) = 1 - (1-x)^{\frac{1}{n-1}}$ and $\tilde{g}_i(x) = \Gamma \cdot \left(-1+ (1-x)^{\frac{n}{n-1}} \right)+1$.

\end{claim*}

\begin{proof}
Let $h(x) \eqdef \Gamma \cdot \left( -(1-x) \ln(1-x)-x\right)+1$.
We first calculate the derivatives of $\tilde{g}_i(x)$ and $h(x)$:
\begin{itemize}
    \item $\tilde{g}_i'(x) = \Gamma \cdot \left( n \left( -1 + \frac{n}{n-1} (1-x)^{\frac{1}{n-1}} \right) - \frac{n}{n-1} (1-x)^{\frac{1}{n-1}} \right) = \Gamma \cdot n \cdot \left( (1-x)^{\frac{1}{n-1}} - 1 \right)$
    \item $h'(x) = \Gamma \cdot ( \ln(1-x) + \frac{1-x}{1-x} - 1) = \Gamma \cdot \ln(1-x)$.
\end{itemize}
Let $I(x) \eqdef \Gamma \cdot \int_0^x \frac{\tilde{p}_i(y)}{\tilde{g}_i(y)} \dd y - \ln h(x) + \ln \tilde{g}_i(x)$. It suffices to show that $I(x) \ge 0$. We calculate the derivative of $I(x)$.
\begin{align*}
I'(x) & = \Gamma \cdot \frac{\tilde{p}_i(x)}{\tilde{g}_i(x)} - \frac{h'(x)}{h(x)} + \frac{\tilde{g}_i'(x)}{\tilde{g}_i(x)} \\
& = \frac{\Gamma \left(1-(1-x)^{\frac{1}{n-1}} \right) + \Gamma \cdot n \cdot \left( (1-x)^{\frac{1}{n-1}} - 1 \right)}{\tilde{g}_i(x)} - \frac{\Gamma \cdot \ln(x)}{h(x)}\\
& = \Gamma \cdot \left( \frac{(n-1) \left((1-x)^{\frac{1}{n-1}} - 1 \right) }{ \tilde{g}_i(x)} - \frac{\ln(1-x)}{h(x)} \right) \\
& = \Gamma \cdot \Bigg( \frac{(n-1) \left((1-x)^{\frac{1}{n-1}} - 1 \right) \cdot \left( \Gamma \cdot \left( -(1-x) \ln(1-x)-x\right)+1 \right)}{\tilde{g}_i(x) \cdot h(x)} \\
& \phantom{= \Gamma \cdot \Bigg(} - \frac{   \ln(1-x) \cdot \left( \Gamma \cdot \left( n \left( 1-x -(1-x)^{\frac{n}{n-1}}\right) -1+ (1-x)^{\frac{n}{n-1}} \right)+1 \right)}{\tilde{g}_i(x) \cdot h(x)} \Bigg) \\
& = \frac{\Gamma \cdot \left( 1-\Gamma x\right) \cdot \left( -\ln(1-x) + (n-1) ((1-x)^{\frac{1}{n-1}}-1) \right)}{\tilde{g}_i(x) \cdot h(x)} 
\end{align*}
Furthermore, by applying $e^z \ge z + 1$ to $z = \frac{\ln(1-x)}{n-1}$, we have
\[
(n-1) ((1-x)^{\frac{1}{n-1}}-1) = (n-1) (e^{\frac{\ln(1-x)}{n-1}}-1) \ge (n-1) \cdot \frac{\ln(1-x)}{n-1} = \ln(1-x).
\]
Hence, $I'(x) \ge 0$. As a result, we have $I(x) \ge I(0) = \Gamma \cdot 0 - \ln h(0) + \ln \tilde{g}_i(0) = 0$, which concludes the proof of the claim.
\end{proof}

\subsection{Proof of Lemma~\ref{lem:integral}}
Recall the statement of the lemma. Here $\Gamma=\frac{\ln \alpha+1}{\ln \alpha+1-\alpha} \approx 0.725$ and $\alpha \approx 0.211 $ is the unique solution of the following equation on $(0,1)$
\[
\int_{\alpha}^{1}  \frac{\ln \alpha+1}{ (\ln \alpha+1) \left( - x \cdot \ln x + x\right) -\alpha} \dd x +\frac{1}{\ln \alpha}=0.
\]
\begin{lemma*}
Suppose functions $\tilde{p}, \tilde{g}:[0,1] \to [0,1]$ satisfies that 
\[
\tilde{g}(x) \ge \Gamma \cdot \left( - (1-x) \cdot \ln(1-x) \cdot (1-\tilde{p}(x)) - x \right) + 1.
\]
Then
\[
\int_0^{1}  \frac{\Gamma}{\tilde{g}(x)  \exp \left( \Gamma \cdot \int_0^{x} \frac{ \tilde{p}(y)}{\tilde{g}(y)} \dd y\right)} \dd x \le 1.
\]
\end{lemma*}
\begin{proof}
Fix arbitrary functions $\tilde{p}(x),\tilde{g}(x)$ that satisfy the stated condition. Define $G:[0,1] \to [0,\infty)$ as the following
\[
G(z) \eqdef \int_z^{1}  \frac{\Gamma}{\tilde{g}(x)  \exp \left( \Gamma \cdot \int_z^{x} \frac{ \tilde{p}(y)}{\tilde{g}(y)} \dd y\right)} \dd x.
\]
Then it is equivalent to prove that $G(0) \le 1$. Taking the derivative of $G(z)$, we have
\begin{multline}
\label{eqn:G}
G'(z) = -\frac{\Gamma}{\tilde{g}(z)} + \int_z^{1}  \frac{\Gamma}{\tilde{g}(x)  \exp \left( \Gamma \cdot \int_z^{x} \frac{ \tilde{p}(y)}{\tilde{g}(y)} \dd y\right)} \cdot \frac{\Gamma \cdot \tilde{p}(z)}{\tilde{g}(z)} \dd x = \frac{\Gamma}{\tilde{g}(z)} \cdot (\tilde{p}(z) \cdot G(z) - 1) \\
\ge \min \left( \frac{-\Gamma}{\Gamma \cdot \left( - (1-z) \cdot \ln(1-z) - z \right) +1}, \frac{-\Gamma \cdot (1-G(z))}{1-\Gamma \cdot z} \right),
\end{multline}
where the inequality follows from the stated condition of $\tilde{p}(z),\tilde{g}(z)$ and the fact that the minimum must be achieved when $\tilde{p}(z) = 0$ or $1$.

Next, we show that for any continuous function $G:[0,1] \to [0,\infty)$ satisfying equation~\eqref{eqn:G} and $G(1) = 0$, we have $G(0) \le 1$. 

Define two auxiliary functions.
\begin{itemize}
\item $H(z) \eqdef \frac{\Gamma \left( -(1-z) \cdot \ln(1-z)\right)}{\Gamma \cdot \left( -(1-z) \cdot \ln(1-z) -z\right)+1}$. This is the function with $\frac{-\Gamma}{\Gamma \cdot \left( - (1-z) \cdot \ln(1-z) - z \right) +1} = \frac{-\Gamma \cdot (1-H(z))}{1-\Gamma \cdot z}$. 
Specifically, by equation~\eqref{eqn:G}, we shall use $G'(z) \ge \frac{-\Gamma \cdot (1-G(z))}{1-\Gamma \cdot z} $ when $G(z) \le H(z)$; and use $G'(z) \ge \frac{-\Gamma}{\Gamma \cdot \left( - (1-z) \cdot \ln(1-z) - z \right) +1} $ when $G(z) > H(z)$.
\item $K(z) \eqdef \frac{\Gamma \cdot (1-z)}{1-\Gamma \cdot z}$ for $z \in [0,1]$. This is the solution to $K'(z) = \frac{-\Gamma \cdot (1-K(z))}{1-\Gamma \cdot z}$. 
\end{itemize}

Let $z_{1} $ be the root of $K(z) = H(z)$ on $[0,1)$. We have
\begin{align*}
\frac{\Gamma \left( -(1-z_1) \cdot \ln(1-z_1) \right)}{\Gamma \cdot \left( -(1-z_1) \cdot \ln(1-z_1) -z_1 \right)+1} = \frac{\Gamma \cdot (1-z_1)}{1-\Gamma \cdot z_1} & \Longrightarrow (1-\Gamma) \cdot \ln(1-z_1)+ 1 - \Gamma \cdot z_1=0 \\
& \Longleftrightarrow \Gamma = \frac{\ln (1-z_1) + 1}{\ln (1-z_1) + z_1}
\end{align*}
Recall that the constant $\Gamma = \frac{\ln \alpha+1}{\ln \alpha+1-\alpha}$, we have $z_1 = 1-\alpha \approx 0.789$. 
Moreover, for $z \in [0,z_1)$, $H(z) < K(z)$; for $z \in (z_1,1)$, we have $H(z) > K(z)$. Refer to Figure~\ref{fig:H_K}.

\begin{figure}[ht]
    \centering
    \includegraphics[width=1.00\textwidth]{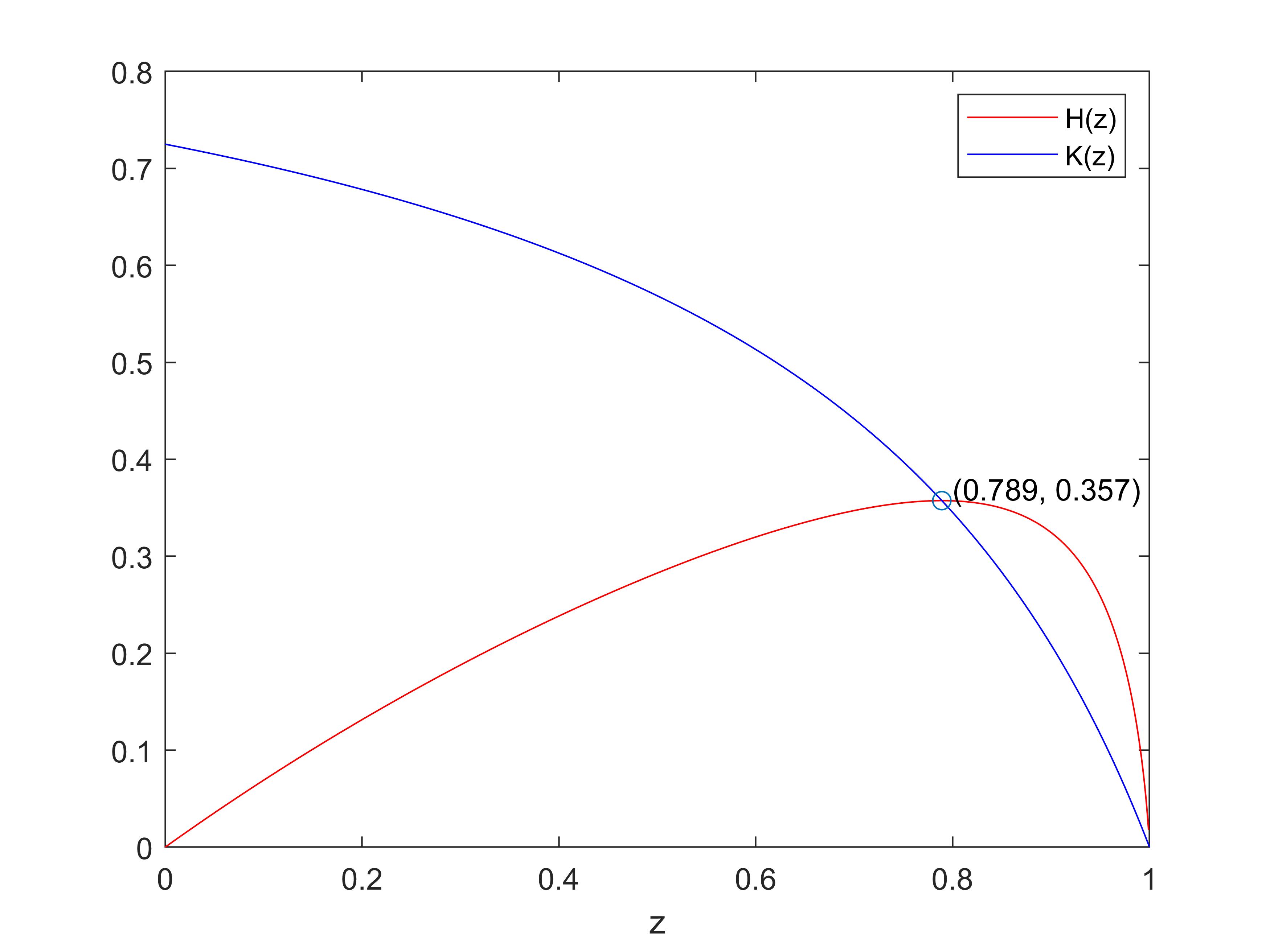}
    \caption{Plots of $H(z)$ and $K(z)$}
    \label{fig:H_K}
\end{figure}

Let $z_0 \eqdef \inf \left\{z \mid G(z) \le H(z), z \in [0,1] \right\}$. Note that $z_0$ is well-defined since $G(1) = 1 = H(1)$.

We claim that $z_{0} \le z_{1}$. It suffices to prove $G(z_{1}) \le H(z_{1})$ and we prove it by contradiction. 

Suppose otherwise $G(z_{1}) > H(z_{1})$. Notice that for sufficiently small $\epsilon > 0$, it holds that $G(z) < H(z)$ for $z \in [1-\epsilon,1)$. Because for $z \in [1-\epsilon,1)$, 
\begin{align*}
G'(z) - H'(z) & = G'(z)- \frac{\Gamma \cdot \left( (1-\Gamma) \cdot \ln(1-z) - \Gamma \cdot z+1\right)}{\left( \Gamma \cdot (-(1-z) \cdot \ln(1-z)-z)+1\right)^{2}} \\
& \ge \min \left( \frac{-\Gamma}{\Gamma \cdot \left( - (1-z) \cdot \ln(1-z) - z \right) +1}, \frac{-\Gamma \cdot (1-G(z))}{1-\Gamma \cdot z}\right) \\
& \phantom{=} -\frac{\Gamma \cdot \left( (1-\Gamma) \cdot \ln(1-z) - \Gamma \cdot z+1\right)}{\left( \Gamma \cdot (-(1-z) \cdot \ln(1-z)-z)+1\right)^{2}} \\
& \ge \frac{-\Gamma}{1- \Gamma \cdot z} -\frac{\Gamma \cdot \left( (1-\Gamma) \cdot \ln(1-z) - \Gamma \cdot z+1\right)}{\left( \Gamma \cdot (-(1-z) \cdot \ln(1-z)-z)+1\right)^{2}} > 0.
\end{align*}
The first inequality holds by equation~\eqref{eqn:G}; 
the second inequality holds by the fact that both
$\frac{-\Gamma}{\Gamma \cdot \left( - (1-z) \cdot \ln(1-z) - z \right) +1}$ and $\frac{-\Gamma \cdot (1-G(z))}{1-\Gamma \cdot z} $ are at least $\frac{-\Gamma}{1 - \Gamma z}$;
the last inequality holds since 
\[
\lim_{z \to 1^{-}} \left(\frac{-\Gamma}{1- \Gamma \cdot z} -\frac{\Gamma \cdot \left( (1-\Gamma) \cdot \ln(1-z) - \Gamma \cdot z+1\right)}{\left( \Gamma \cdot (-(1-z) \cdot \ln(1-z)-z)+1\right)^{2}} \right) = +\infty.
\]

Then for $z \in (1-\epsilon,1)$,
\[
G(z) -H(z) = G(1) - H(1)  - \int_{z}^{1} \left(G'(t) - H'(t) \right) \dd t = - \int_{z}^{1} \left(G'(t) - H'(t) \right) \dd t < 0.
\]

Let $z_{2} = \inf\{z \mid z \in (0,1], G(y) \le H(y) \text{ for } \forall y \in [z,1) \}$. According to the derivation above, we have $z_2 \le 1-\epsilon$. Moreover, since both $G,H$ are continuous functions, we have $G(z_2) = H(z_2)$.
Since $G(z_{1}) > H(z_{1})$, we have $z_{2} \in (z_{1},1-\epsilon]$. 
According to the definition of $z_2$ and $H(z)$, for every $z \in [z_2, 1]$,
\begin{align*}
& G'(z) \ge \min \left( \frac{-\Gamma}{\Gamma \cdot \left( - (1-z) \cdot \ln(1-z) - z \right) +1}, \frac{-\Gamma \cdot (1-G(z))}{1-\Gamma \cdot z}\right) = \frac{-\Gamma \cdot (1-G(z))}{1-\Gamma \cdot z} \\
& \Longrightarrow \left( (1 - \Gamma z) \cdot G(z) \right)' = (1-\Gamma z) \cdot G'(z) - \Gamma \cdot G(z) \ge - \Gamma \\
& \Longrightarrow (1-\Gamma) \cdot G(1) - (1-\Gamma z) \cdot G(z) \ge - \Gamma \cdot (1-z) \\
& \Longrightarrow G(z) \le \frac{\Gamma \cdot (1-z)}{1-\Gamma z} = K(z).
\end{align*}
Moreover, $K(z) < H(z)$ for $z \in (z_{1},1)$ by the definition of $z_1$. This implies that $G(z) < H(z)$ for $z \in [z_{2},1)$, that contradicts $G(z_{2})=H(z_{2})$. Therefore, we have $z_{0} \le z_{1} $.

For $z \in [0, z_{0})$, we have $G(z) > H(z)$, and, hence
\[
G'(z) \ge \min \left( \frac{-\Gamma}{\Gamma \cdot \left( - (1-z) \cdot \ln(1-z) - z \right) +1
}, \frac{-\Gamma \cdot (1-G(z))}{1-\Gamma \cdot z}\right) = \frac{-\Gamma}{\Gamma \cdot \left( - (1-z) \cdot \ln(1-z) - z \right) +1}.
\]
Thus,
\begin{equation}
\label{eqn:G0}
G(0) = G(z_{0}) - \int_{0}^{z_{0}} G'(z) \dd z  \le H(z_{0}) + \int_{0}^{z_{0}}  \frac{\Gamma}{\Gamma \cdot \left( - (1-z) \cdot \ln(1-z) - z \right) +1} \dd z 
\end{equation}

Let $L(z) \eqdef H(z) + \int_{0}^{z}  \frac{\Gamma}{\Gamma \cdot \left( - (1-x) \cdot \ln(1-x) - x \right) +1} \dd x$.
We verify that $L'(z) \ge 0$ for $z \in [0,z_{1}]$. Indeed,
\begin{align*}
 L'(z) & = \frac{\Gamma \cdot \left( (1-\Gamma) \cdot \ln(1-z) - \Gamma \cdot z+1\right)}{\left( \Gamma \cdot (-(1-z) \cdot \ln(1-z)-z)+1\right)^{2}} +\frac{\Gamma}{\Gamma \cdot \left( - (1-z) \cdot \ln(1-z) - z \right) +1}  \\
 & =  \frac{\Gamma \cdot \left( (\Gamma \cdot z-2 \cdot \Gamma+1) \cdot \ln(1-z) - 2 \cdot \Gamma \cdot z+2\right)}{\left( \Gamma \cdot (-(1-z) \cdot \ln(1-z)-z)+1\right)^{2}}.
\end{align*}
Denote the numerator by $l(z) = (\Gamma \cdot z-2 \cdot \Gamma+1) \cdot \ln(1-z) - 2 \cdot \Gamma \cdot z+2$. Then $l'(z) = \Gamma \cdot \ln(1-z) + \frac{\Gamma \cdot z-1}{1-z} \le 0$. In order to prove $L'(z) \ge 0$ for $z \in [0,z_{1}]$, it suffices to verify that:
\begin{multline*}
l(z_{1}) =  (\Gamma \cdot z_{1}-2 \cdot \Gamma+1) \cdot \ln(1-z_{1}) - 2 \cdot \Gamma \cdot z_{1}+2 \\
= (\Gamma \cdot z_{1}-2 \cdot \Gamma+1) \cdot \left( - \frac{1-\Gamma \cdot z_{1}}{1-\Gamma} \right) - 2 \cdot \Gamma \cdot z_{1}+2 = \frac{(1- \Gamma \cdot z_{1})^{2}}{1 - \Gamma} \ge 0.
\end{multline*}

Finally, by the monotonicity of $L(z)$ on $[0,z_1]$ and the fact that $z_0 \le z_1$, we have
\begin{align*}
G(0) \overset{\eqref{eqn:G0}}{\le} L(z_{0}) \le L(z_{1}) & =  H(z_{1})+\int_{0}^{z_{1}}  \frac{\Gamma}{\Gamma \cdot \left( - (1-z) \cdot \ln(1-z) - z \right) +1} \dd z \\
& = K(z_{1})+\int_{0}^{z_{1}}  \frac{\Gamma}{\Gamma \cdot \left( - (1-z) \cdot \ln(1-z) - z \right) +1} \dd z \\
& = \frac{\Gamma \cdot (1-z_{1})}{1-\Gamma \cdot z_{1}} +\int_{0}^{z_{1}}  \frac{\Gamma}{\Gamma \cdot \left( - (1-z) \cdot \ln(1-z) - z \right) +1} \dd z  \\
& = \frac{\ln \alpha +1}{\ln \alpha} +\int_{0}^{1-\alpha}  \frac{\frac{\ln \alpha +1}{\ln \alpha +1- \alpha}}{\frac{\ln \alpha+1}{\ln \alpha+1-\alpha} \cdot \left( - (1-z) \cdot \ln(1-z) - z \right) +1} \dd z \\
& = 1+\frac{1}{\ln \alpha}+ \int_{\alpha}^{1}  \frac{\frac{\ln \alpha +1}{\ln \alpha+1-\alpha}}{\frac{\ln \alpha +1}{\ln \alpha+1-\alpha} \cdot \left( - y \cdot \ln y +y-1 \right) +1} \dd y \tag{let $y=1-z$} \\
& = 1+\frac{1}{\ln \alpha}+ \int_{\alpha}^{1}  \frac{\ln \alpha +1}{\left(\ln \alpha+1 \right) \cdot \left( - y \cdot \ln y + y \right) - \alpha} \dd y = 1,
\end{align*}
where the last equality follows from the definition of constant $\alpha$.

\end{proof} 

\end{document}